\documentclass[preprint,12pt]{elsarticle}

\usepackage{amssymb}

\journal{Nuclear Physics B}

\usepackage[ruled,vlined,algosection,linesnumbered]{algorithm2e}
\usepackage{graphicx}
\usepackage{complexity}
\usepackage{amssymb}
\usepackage{savesym}
\usepackage{amsthm}
\usepackage{amsmath}

\newtheorem{theorem}{Theorem}
\newtheorem{lemma}{Lemma}
\newtheorem{definition}{Definition}
\newtheorem{corollary}{Corollary}
\newtheorem{observation}{Observation}
\newtheorem{reduction rule}{Reduction Rule}[section]
\newtheorem{claim}{Claim}[section]

\newtheorem{branching rule}{Branching Rule}[section]

\usepackage{tikz}
\usetikzlibrary{positioning}
\usetikzlibrary{arrows}
\usetikzlibrary{shapes}
\usetikzlibrary{matrix}
\usepackage{graphicx}

\newenvironment{lemma-l}[1]{\noindent {\bf Lemma~#1.~}\em }{\smallskip}
\newenvironment{theorem-t}[1]{\noindent {\bf Theorem~#1.~}\em }{\smallskip}
\newtheorem{obs}[theorem]{Observation}

\usepackage{amsmath,amssymb}
\newtheorem{Reduction Rule}{Reduction Rule}

\newcommand{\containment}{\NP~$\subseteq$~\coNP/poly\xspace}

\newcommand{\Oh}{\mathcal{O}}

\newcommand{\cdp}{\textsc{cd-Partization}}

\usepackage{xcolor}
\usepackage{boxedminipage}

\newcommand{\tw}{{\mathbf{tw}}}
\newcommand{\inc}{{\mathbf{inc}}}
\newcommand{\adj}{{\mathbf{adj}}}

\newcommand{\defdecproblem}[3]{
	\vspace{3mm}
	\noindent\fbox{
		\begin{minipage}{.95\textwidth}
			\begin{tabular*}{\textwidth}{@{\extracolsep{\fill}}lr} \textsc{#1}  &  \\ \end{tabular*}
			{\bf{Input:}} #2  \\
			{\bf{Question:}} #3
		\end{minipage}
	}
	\vspace{2mm}
}

\begin{document}

\begin{frontmatter}



\title{Parameterized and Exact Algorithms for Class Domination Coloring \tnoteref{conf-version}}
\tnotetext[conf-version]{A preliminary version of this paper appeared in the proceedings of $43^{rd}$ International Conference on Current Trends in Theory and Practice of Computer Science (SOFSEM 2016).}

\author{R. Krithika}
\ead{krithika@iitpkd.ac.in}
\address{Indian Institute of Technology Palakkad, Palakkad, India}

\author{Ashutosh Rai}
\ead{ashutosh.rai@maths.iitd.ac.in}
\address{Indian Institute of Technology Delhi, Delhi, India}

\author{Saket Saurabh\fnref{label2}}
\ead{saket@imsc.res.in}
\fntext[label2]{The research leading to these results has received funding from the European Research Council under the European Union's Seventh Framework Programme (FP7/2007-2013) / ERC grant agreement no. 306992}
\address{The Institute of Mathematical Sciences, HBNI, Chennai, India and University of Bergen, Bergen, Norway}

\author{Prafullkumar Tale\fnref{label3}}
\ead{prafullkumar.tale@cispa.saarland}
\address{CISPA Helmholtz Center for Information Security, Saarbr$\ddot{\text{u}}$cken, Germany}

\fntext[label3]{This research is a part of a project that has received funding from the European Research Council (ERC) under the European Union's Horizon $2020$ research and innovation programme under grant agreement SYSTEMATICGRAPH (No. $725978$).}

\begin{abstract}
A class domination coloring (also called cd-Coloring or dominated coloring) of a graph is a proper coloring in which every color class is contained in the neighbourhood of some vertex. 
The minimum number of colors required for any cd-coloring of $G$, denoted by $\chi_{cd}(G)$, is called the class domination chromatic number (cd-chromatic number) of $G$.
In this work, we consider two problems associated with the cd-coloring of a graph in the context of exact exponential-time algorithms and parameterized complexity.
(1) Given a graph $G$ on $n$ vertices, find its cd-chromatic number.
(2) Given a graph $G$ and integers $k$ and $q$, can we delete at most $k$ vertices such that the cd-chromatic number of the resulting graph is at most $q$?
For the first problem, we give an exact algorithm with running time $\Oh(2^n n^4 \log n)$.
Also, we show that the problem is \FPT\ with respect to the number $q$ of colors as the parameter on chordal graphs.
On graphs of girth at least 5, we show that the problem also admits a kernel with $\Oh(q^3)$ vertices.
For the second (deletion) problem, we show \NP-hardness for each $q \geq 2$.
Further, on split graphs, we show that the problem is \NP-hard if $q$ is a part of the input and \FPT\ with respect to $k$ and $q$ as combined parameters.
As recognizing graphs with cd-chromatic number at most $q$ is \NP-hard in general for $q \geq 4$, the deletion problem is unlikely to be \FPT\ when parameterized by the size of the deletion set on general graphs.
We show fixed parameter tractability for $q \in \{2,3\}$ using the known algorithms for finding a vertex cover and an odd cycle transversal as subroutines. 
\end{abstract}







\end{frontmatter}



	\section{Introduction}
Graph coloring is a classical problem in the fields of combinatorics and algorithm design. A {\em proper coloring} of a graph is an assignment of colors to its vertices such that no two adjacent vertices receive the same color. 
Equivalently, a proper coloring is a partition of the vertex set into independent sets.
In this context, these independent sets are also called {\em color classes}.
A proper coloring of a graph $G$ using $q$ colors is called a {\em $q$-coloring} of $G$ and the minimum number of colors required in a proper coloring is called as the {\em chromatic number} of $G$. Determining the chromatic number of a graph is a classical \NP-hard problem. This problem has been widely investigated in the areas of exact algorithms~\cite{BjorklundHK09,GaspersKLT09,GaspersL12,Kratsch08,LAWLER197666,RooijB11}, approximation algorithms~\cite{BlumK97,GuhaK99,Kim10,LenzenW10}, and parameterized algorithms~\cite{AlberBFKN02,AlonG09,Cai03a,DowneyFMR08}.  Further, variants of the graph coloring like \textsc{Edge-Chromatic Number}, \textsc{Achromatic Number}, \textsc{$b$-Chromatic Number}, \textsc{Total Chromatic Number}, \textsc{Dominator Coloring} and \textsc{Class Domination Coloring} have also been well studied~\cite{Gera2006,Gera2007,dom-col-3}. 

In this work, we initiate the study of \textsc{Class Domination Coloring} (also called \textsc{cd-Coloring} or \textsc{Dominated Coloring}) in the realm of parameterized complexity and exact exponential time algorithms.
A {\em cd-coloring} is a proper coloring of the graph in which every color class is contained in the neighbourhood of some vertex. 
See Figure~\ref{fig:cd-col} for an example.
The minimum number of colors needed in any cd-coloring of $G$ is called the {\em class domination chromatic number} or {\em cd-chromatic number} of $G$ and is denoted by $\chi_{cd}(G)$. Also, $G$ is said to be $q$-cd-colorable if $\chi_{cd}(G) \leq q$. The \textsc{cd-Coloring} problem is formally defined as follows. 
	
\defdecproblem{cd-Coloring}{A graph $G$ and a positive integer $q$.}{Is $\chi_{cd}(G) \leq q$?}
	
\textsc{cd-Coloring} is \NP-complete for $q \geq 4$ and polynomial-time solvable for $q \leq 3$ \cite{caldam16}. A characterization of graphs that admit 3-cd-colorings is also known \cite{caldam16}. \textsc{cd-Coloring} has also been studied on many restricted graph classes like split graphs, $P_4$-free graphs~\cite{caldam16} and middle and central graphs of $K_{1,n}$, $C_n$ and $P_n$~\cite{vs10}. 
See also \cite{abid2018dominated, merouane2015dominated,  shalu2017lower, shalu2020complexity, choopani2018dominated}.
	
We study this problem in the context of exact exponential-time algorithms and parameterized complexity.
The field of exact algorithms typically deals with designing algorithms for \NP-hard problems that are faster than brute-force search while the goal in parameterized complexity is to provide efficient algorithms for \NP-complete problems by switching from the classical view of single-variate measure of the running time to a multi-variate one. In parameterized complexity, we consider instances $(I,k)$ of a parameterized problem $\Pi \subseteq \Sigma^* \times \mathbb{N}$, where $\Sigma$ is a finite alphabet. Algorithms in this area have running times of the form $f(k)|I|^{\Oh(1)}$, where $k$ is an integer measuring some part of the instance. This integer $k$ is called the {\em parameter}, and a problem that admits such an algorithm is said to be {\em fixed-parameter tractable} (\FPT). In most of the cases, the solution size is taken to be the parameter, which means that this approach results in efficient (polynomial-time) algorithms when the solution is of small size. A \emph{kernelization} algorithm for a parameterized problem $\Pi$ is a polynomial time procedure which takes as input an instance $(x,k)$ of $\Pi$ and returns an instance $(x',k')$ such that $(x,k) \in \Pi$ if and only if $(x',k')\in \Pi$ and $|x'| \leq h(k)$ and $k' \leq g(k)$, for some computable functions $h,g$. The returned instance is called a {\it kernel} and $h(k)+g(k)$ is its {\it size}. We say that $\Pi$ admits a {\em polynomial kernel} if $h$ and $g$ are polynomials. For more background on parameterized complexity, we refer the reader to the monographs \cite{fpt-book,ParameterizedComplexityBook,FlumGroheBook,RN}.
	
\begin{figure}[t]
	\centering
	  \includegraphics[width=0.5\textwidth]{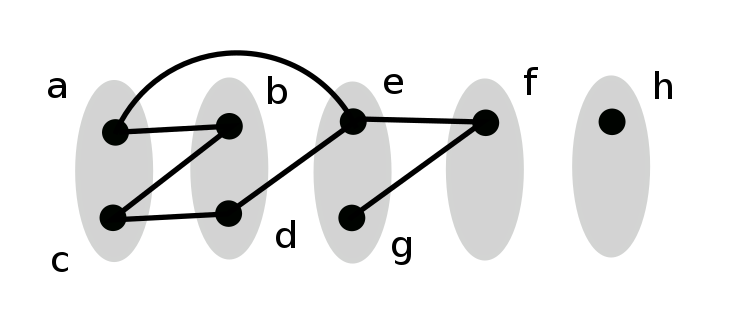}
	\caption{An example of a cd-Coloring of a graph}
	\label{fig:cd-col}
\end{figure}

We first observe that parameterizing \textsc{cd-Coloring} by the solution size (which is the number of colors) does not help in designing efficient algorithms as the problem is para-\NP-hard (\NP-hard even when the parameter is a constant).
Hence, this problem is unlikely to be \FPT\ when parameterized by the solution size. 
Then, we describe an $\Oh(2^n n^4 \log n)$-time algorithm for finding the cd-chromatic number of a graph using polynomial method.
Next, we show that \textsc{cd-Coloring} is \FPT\ when parameterized by the number of colors and the treewidth of the input graph. Further, we show that the problem is \FPT\ when parameterized by the number of colors on chordal graphs. Kaminski and Lozin \cite{lozin2007coloring} showed that determining if a graph of girth at least $g$ admits a proper coloring with at most $q$ colors or not is \NP-complete for any fixed $q \ge 3$ and $g \ge 3$. In particular, \textsc{Chromatic Number} is para-\NP-hard for graphs of girth at least 5. In contrast, we show that \textsc{cd-Coloring} is \FPT\ on this graph class and admits a kernel with $\Oh(q^3)$ vertices.

On a graph $G$ that is not $q$-cd-colorable, a natural optimization question is to check if we can delete at most $k$ vertices from $G$ such that the cd-chromatic number of the resultant graph is at most $q$. We define this problem as follows. 
	
	\defdecproblem{\cdp}{Graph $G$, integers $k$ and $q$}{Does there exist $S \subseteq V(G)$, $|S| \leq k$, such that $\chi_{cd}(G-S)\leq q$?}
	
If $q$ is fixed, then we refer to the problem as \textsc{$q$-\cdp}. Once again, from parameterized complexity point of view, this question is not interesting on general graphs for values of $q$ greater than three, as in those cases, an \FPT\ algorithm with deletion set (solution) size as the parameter is a polynomial-time recognition algorithm for $q$-cd-colorable graphs. Hence, the deletion question is interesting only on graphs where the recognition problem is polynomial-time solvable. We show that \textsc{$q$-\cdp} is \NP-complete for each $q \geq 2$, and that for $q \in \{2,3\}$, the problem is \FPT\ with respect to the solution size as the parameter. Our algorithms use the known parameterized algorithms for finding a vertex cover and an odd cycle transversal of a graph as subroutines. We also show that \cdp\ remains \NP-complete on split graphs and is \FPT\ when parameterized by the number of colors and solution size.
	
	\section{Preliminaries}
	
	The set of integers $\{1,2,\ldots,k\}$ is denoted by $[k]$. All graphs considered in this paper are finite, undirected and simple. For the terms which are not explicitly defined here, we use standard notations from \cite{diestel2000graph}. For a graph $G$, its vertex set is denoted by $V(G)$ and its edge set is denoted by $E(G)$. For a vertex $v \in V(G)$, its (open) {\em neighbourhood} $N_G(v)$ is the set of all vertices adjacent to it and its {\em closed neighborhood} is the set $N_G(v) \cup \{v\}$. We omit the subscript in the notation for neighbourhood if the graph under consideration is clear from the context. The degree of a vertex $v$ is the size of its open neighborhood. 
	
	For a set $S\subseteq V(G)$, the {\it subgraph of $G$ induced by $S$}, denoted by $G[S]$, is defined as the subgraph of $G$ with vertex set $S$ and edge set $\{(u,v) \in E(G) :u,v\in S\}$. The subgraph of $G$ obtained after deleting $S$ (and the edges incident on it) is denoted as $G- S$. The {\em girth} of a graph is the length of a smallest cycle. A set $D \subseteq V(G)$ is said to be a {\em dominating set} of $G$ if every vertex in $V(G) \setminus D$ is adjacent to some vertex in $D$. 
	
	A {\em proper coloring} of $G$ with $q$ colors is a function $f: V(G) \rightarrow [q]$ such that for all $(u,v) \in E(G)$, $f(u) \neq f(v)$. For a proper coloring $f$ of $G$ with $q$ colors and $i \in [q]$, $f^{-1} (i) \subseteq V(G)$ is called a {\em color class} in the coloring $f$. The {\em chromatic number} $\chi(G)$ of $G$ is the minimum number of colors required in a proper coloring of $G$. A {\em clique} is a graph which has an edge between every pair of vertices. The {\em clique number} $\omega(G)$ of $G$ is the size of a largest clique which is a subgraph of $G$. A {\em vertex cover} is a set of vertices that contains at least one endpoint of every edge in the graph. An {\em independent set} is a set of pairwise nonadjacent vertices. A graph is said to be a {\em bipartite graph} if its vertex set can be partitioned into two independent sets. An {\em odd cycle transversal} is a set of vertices whose deletion from the graph results in a bipartite graph. A {\em tree-decomposition} of a graph $G$ is a pair 
	$(\mathbb{T},\mathcal{ X}=\{X_{t}\}_{t\in V({\mathbb T})})$ such that
	\begin{itemize}
		\item $\bigcup_{t\in V(\mathbb{T})}{X_t}=V(G)$,
		\item for every edge $(x,y)\in E(G)$ there is a $t\in V(\mathbb{T})$ such that  $\{x,y\}\subseteq X_{t}$, and 
		\item for every  vertex $v\in V(G)$ the subgraph of $\mathbb{T}$ induced by the set  $\{t\mid v\in X_{t}\}$ is connected.
	\end{itemize}
	
	\noindent The {\em width} of a tree decomposition is $\max_{t\in V(\mathbb{T})} |X_t| -1$ and the {\em treewidth} of $G$, denoted by $\tw(G)$, is the minimum width over all tree decompositions of $G$. The syntax of {\em Monadic Second Order Logic (MSO)} of graphs includes the logical connectives $\vee$, $\wedge$, $\neg$, $\Rightarrow$, $\Leftrightarrow$, variables for vertices, edges, sets of vertices, sets of edges, the quantifiers $\forall$, $\exists$ that can be applied to these variables and the following five binary relations.
	\begin{itemize}
		\item $u \in U$ where $u$ is a vertex variable and $U$ is a vertex set variable;
		\item $e \in F$ where $e$ is an edge variable and $F$ is an edge set variable;
		\item $\inc(e,u)$, where $e$ is an edge variable, $u$ is a vertex variable, and the interpretation is that the edge $e$ is incident with the vertex $u$;
		\item $\adj(u,v)$, where $u$ and $v$ are vertex variables and the interpretation is that $u$ and $v$ are adjacent;
		\item equality of variables representing vertices, edges, sets of vertices, and sets of edges.
	\end{itemize}
	
	
	
	\noindent For an MSO formula $\phi$, $||\phi||$ denotes the length of its encoding as a string.
	
	\begin{theorem}[Courcelle's theorem, \cite{Courcelle90,Courcelle92}] \label{thm:courcelle} Let $\phi$ be a graph property that is expressible in MSO. Suppose $G$ is a graph on $n$ vertices with treewidth $tw$ equipped with the evaluation of all the free variables of $\phi$. Then, there is an algorithm that verifies whether $\phi$ is satisfied in $G$ in $f(||\phi||, tw) \cdot n$ time for some computable function $f$.
	\end{theorem}
	
	\noindent We end the preliminaries section with following simple observations.
	
	\begin{observation} \label{obs:graph-connected} If $G_1,\ldots,G_l$ are the connected components of $G$, then $\chi_{cd}(G) = \sum_{i=1}^{l} \chi_{cd}(G_i)$.
	\end{observation}  
	
	
	\begin{observation} \label{obs:color-dom-set} If $G$ is $q$-cd-colorable, then $G$ has a dominating set of size at most $q$.
	\end{observation}
	

	\section{Exact Algorithm for cd-Chromatic Number}
	Let $G$ denote the input graph on $n$ vertices. Given a coloring of $V(G)$, we can check in polynomial time whether it is a cd-coloring or not. Therefore, to compute $\chi_{cd}(G)$, we can iterate over all possible colorings of $V(G)$ with at most $n$ colors and return a valid cd-coloring that uses the minimum number of colors. This brute force algorithm runs in $2^{\mathcal{O}(n \log n)}$ time. In this section we present an algorithm which runs in $\mathcal{O}(2^{n}n^4 \log n)$ time. The idea for this algorithm is inspired by an exact algorithm for \textsc{$b$-Chromatic Number} presented in \cite{panolan2015b}. We first list some preliminaries on polynomials and Fast Fourier Transform following the framework of \cite{panolan2015b}.
	
	A binary vector $\phi$ is a finite sequence of bits and $val(\phi)$ denotes the integer $d$ of which $\phi$ is the binary representation. All vectors considered here are binary vectors and are synonymous to binary numbers. Further, they are the binary representations of integers less than $2^n$ and are assumed to consist of $n$ bits. $\phi_1 + \phi_2$ denotes the vector obtained by the bitwise addition of the binary numbers (vectors) $\phi_1$ and $\phi_2$. Let $U= \{ u_1 , u_2,\dots , u_n \}$ denote a universe with a fixed ordering on its elements. The {\em characteristic vector} of a set $S \subseteq U$, denoted by $\psi(S)$, is the vector of length $|U|$ whose $j^{\text{th}}$ bit is $1$ if $u_j \in S$ and $0$ otherwise. The {\em Hamming weight} of a vector $\phi$ is the number of $1$s in $\phi$ and it is denoted by $\mathcal{H}(\phi)$. Observe that $\mathcal{H}(\psi(S)) = |S|$. The Hamming weight of an integer is define as hamming weight of its binary representation. To obtain the claimed running time bound for our exponential-time algorithm, we make use of the algorithm for multiplying polynomials based on the Fast Fourier Transform.
	
	\begin{lemma}[\cite{schonhage1971schnelle}] \label{lemma:fft} Two polynomials of degree at most $d$ over any commutative ring $\mathcal{R}$ can be multiplied using $\mathcal{O}(d \cdot \log d \cdot \log \log d)$ additions and multiplications in $\mathcal{R}$.
	\end{lemma}
	
	\noindent Let $z$ denote an indeterminate variable. We use the monomial $z^{val(\psi(S))}$ to represent the set $S \subseteq U$ and as a natural extension, we use univariate polynomials to represent a family of sets. 
	
	\begin{definition}[Characteristic Polynomial of a Family of Sets] For a family $\mathcal{F} = \{S_1, S_2, \dots , S_q\}$ of subsets of $U$, the characteristic polynomial of $\mathcal{F}$ is defined as $p_\psi(\mathcal{F}) = \sum_{i = 1}^{q} z ^{val(\psi (S_i))}$.
	\end{definition}
	
	\begin{definition}[Representative Polynomial] For a polynomial $p(z) = \sum_{i = 1}^{q} a_i \cdot z ^{i} $, we define its representative polynomial as $ \sum_{i = 1}^{q} b_i \cdot z ^{i} $ where $b_i = 1$ if $a_i \neq 0$ and $b_i = 0$ if $a_i = 0$. 
	\end{definition}
	
	\begin{definition} [Hamming Projection]The Hamming projection of the polynomial $p(z)= \sum_{i = 1}^{q} a_i \cdot z ^{i}$ to the integer $h$ is defined as $\mathcal{H}_{h}(p(z)) := \sum_{i = 1}^{q} b_i \cdot z ^{i} $ where $b_i = a_i$ if $\mathcal{H}(i) = h$ and $b_i = 0$ otherwise. 
	\end{definition}
	
	\noindent Next, for two sets $S_1,S_2 \subseteq U$, we define a modified multiplication operation $(\star)$ of the monomials $z^{\psi(S_1)}$ and $z^{\psi(S_2)}$ in the following way. 
	
	\[
	z^{val(\psi(S_1))} \star z^{val(\psi(S_2))} =
	\begin{cases}
	z^{val(\psi(S_1)) + val(\psi(S_2))} & \text{if } S_1 \cap S_2 = \emptyset \\
	0 & \text{otherwise}
	\end{cases}
	\]

	\noindent For a polynomial function $p(z)$ of $z$ and a positive integer $\ell \ge 2$, we inductively define the polynomial $p(z)^{\ell}$ as $p(z)^{\ell} := p(z)^{\ell - 1} \star p(z)$. Here, coefficients of monomials follow addition and multiplications defined over underlying field. We now describe an algorithm for implementing the $\star$ operation using the standard multiplication operation and the notion of Hamming weights of bit strings associated with exponents. 
	
	\begin{algorithm}[H]
		\KwIn{Two polynomials $q(z), r(z)$ of degree at most $2^n$}
		\KwOut{$q(z) \star r(z)$ }
		Initialize polynomials $t(z)$ and $t'(z)$ to 0\\
		
		\For{ each ordered pair $(i, j) \text{ such that } i + j \le n$}{
			Compute $s_i(z) = \mathcal{H}_i(q(z))$ and $s_j(z) = \mathcal{H}_j(r(z))$\\
			Compute $s_{ij}(z) = s_i(z) * s_j(z)$ using Lemma~\ref{lemma:fft} \label{step:multiply}\\ 
			$t'(z) = t(z) + \mathcal{H}_{i + j}(s_{ij}(z))$\\
			Set $t(z)$ as the representative polynomial of $t'(z)$ \label{step:rep-poly}
		}
		\Return $t(z)$
		\caption{Compute ($\star$) product of two polynomials}
		\label{alg:compute-star}
	\end{algorithm}
	
	\begin{lemma} 
		\label{lemma:star}
		Let $\mathcal{F}_1$ and $\mathcal{F}_2$ be two families of subsets of $U$. Let $\mathcal{F}$ denote the collection $\{S_1 \cup S_2 |\ S_1 \in \mathcal{F}_1, S_2 \in \mathcal{F}_2 \text{ and } S_1 \cap S_2 = \emptyset\}$. Then, $p_\psi(\mathcal{F}_1) \star p_\psi(\mathcal{F}_2)$ computed by Algorithm~\ref{alg:compute-star} is $p_\psi(\mathcal{F})$.
	\end{lemma}
	\begin{proof} 
		Define $q(z)=p_\psi(\mathcal{F}_1)$, $r(z)= p_\psi(\mathcal{F}_2)$ and $t(z)=q(z) \star r(z)$. Let $S_1 \in \mathcal{F}_1$ and $S_2 \in \mathcal{F}_2$ be sets such that $S_1 \cap S_2 = \emptyset$. Define $S=S_1 \cup S_2$ and let $\phi_1, \phi_2$ and $\phi$ be the characteristic vectors of $S_1, S_2$, and $S$ respectively. We claim that the term $z^{val(\phi)}$ is present in $t(z)$. For a vector $\phi$ and an integer $i \in [n]$, let $\phi[i]$ denote the $i^{th}$ bit in $\phi$. As $\phi[i]$ is 1 if and only if exactly one of the two bits $\phi_1[i]$, $\phi_2[i]$ is 1, it follows that there is no carry at any position (and hence no overflow) while adding $\phi_1$ and $\phi_2$. Therefore, $\phi=\phi_1 + \phi_2$ is a binary string of $n$ bits and $\mathcal{H}(\phi) = \mathcal{H}(\phi_1) + \mathcal{H}(\phi_2)$. Now, as $q(z)$ contains $z^{val(\phi_1)}$ and $r(z)$ contains $z^{val(\phi_2)}$, in the execution of Algorithm \ref{alg:compute-star}, for $i = |S_1|$ and $j = |S_2|$, polynomials $s_i(z)$ and $s_j(z)$ contain $z^{val(\phi_1)}$ and $z^{val(\phi_2)}$ respectively. Step~\ref{step:multiply} multiplies $s_i(z)$ and $s_j(z)$ using Fast Fourier Transformation to obtain $s_{ij}(z)$. As $\mathcal{H}(\phi_1)=i$, $\mathcal{H}(\phi_2)=j$ and $\mathcal{H}(\phi_1) + \mathcal{H}(\phi_2) = i + j$, $s_{ij}(z)$ contains the term $z^{val(\phi)}=z^{val(\phi_1) + val(\phi_2)}$. Moreover, $z^{val(\phi)}$ is present in $\mathcal{H}_{i + j}(s_{ij}(z))$ and hence it is a monomial in $t(z)$ as Step~\ref{step:rep-poly} ensures that every monomial in $t(z)$ is of the form $z^d$ for some integer $d$. 
		
		Next, we show that for every monomial $z^d$ in $t(z)$, there is a set $S \in \mathcal{F}$ such that $d=val(\psi(S))$. Let $i$ and $j$ be integers such that $\mathcal{H}_{i + j}(s_{ij}(z))$ contains the term $z^d$. As $t(z)$ was initialized to $0$, $z^d$ was obtained as the product of two terms $z^{d_1}, z^{d_2}$ in $s_i(z)$ and $s_j(z)$ respectively such that $d_1 + d_2 = d$. Let $S_1 \in \mathcal{F}_1$ be the set such that $\psi(S_1)$ is the binary representation of $d_1$. Similarly, let $S_2 \in \mathcal{F}_2$ be the set such that $\psi(S_2)$ is the binary representation of $d_2$. Let $\phi_1$ and $\phi_2$ be the characteristic vectors of $S_1$ and $S_2$ respectively. Then, $|S_1|=i$, $|S_2|=j$ and there is no integer $k$ between $1$ and $n$ such that $\phi_1[k]=\phi_2[k]=1$. Therefore, $S_1 \cap S_2 =\emptyset$ and $z^d=z^{val(S_1 \cup S_2)}$. Hence, the claimed set $S$ is $S_1 \cup S_2$ which is in $\mathcal{F}$ as $S_1 \cap S_2 =\emptyset$. 
	\end{proof}
	
	\begin{corollary}
		\label{cor:running-time} Given a polynomial $p(z)$ of degree at most $2^n$, there is an algorithm that computes $p(z)^{\ell}$ in $\mathcal{O}(2^n n^3 \log n \cdot l)$ time.
	\end{corollary}
	\begin{proof}
		By Lemma~\ref{lemma:fft}, an execution of the Fast Fourier multiplication algorithm takes $\mathcal{O}(2^n n \log n)$ time. As the \textbf{for} loop of Algorithm \ref{alg:compute-star} is executed $n^2$ times, the total time to compute $p(z)^{\ell}$ is $\mathcal{O}(2^n n^3 \log n)$.  
	\end{proof}
	
	\noindent We now prove a result which correlates the existence of a partition of a set with the presence of a monomial in a polynomial associated with it.
	
	\begin{lemma} \label{lemma:generalized-partition} Consider a universe $U$ and a family $\mathcal{F}$ of its subsets with characteristic polynomial $p(z)$. For any $W \subseteq U$, $W$ is the disjoint union of $\ell$ sets from $\mathcal{F}$ if and only if there exists a monomial $z^{val(\psi(W))}$ in $p(z)^{\ell}$.   
	\end{lemma}
	
	\begin{proof} 
		Let $W$ be the disjoint union of $S_1, S_2, \dots, S_{\ell}$ such that $S_i \in \mathcal{F}$ for all $i \in [\ell]$. For any $j \in [n]$, the $j^{\text{th}}$ bit of $\psi(W)$ is 1 if and only if there is exactly one $S_i$ such that $j^{th}$ bit of $\psi(S_i)$ is 1. Thus, $val(\psi(W)) = val(\psi(S_1)) + val(\psi(S_2)) + \dots + val(\psi(S_{\ell}))$. Now, for every $S_i$ there is a term $z^{val(\psi(S_i))}$ in $p(z)$. Further, as the $S_i$'s are pairwise disjoint, the monomial $z^{val(\psi(S_1))} \star z^{val(\psi(S_2))} \star \cdots \star z^{val(\psi(S_{\ell}))}$ which is equal to $z^{val(\psi(W))}$ is present in $p(z)^{\ell}$. 
		
		We prove the converse by induction on $\ell$. For $\ell=1$, the statement is vacuously true and for $\ell=2$, the claim holds from the proof of Lemma \ref{lemma:star}. Assume that the claim holds for all the integers which are smaller than $\ell$, that is, if there exists a monomial $z^{val(\psi(W))}$ in $p(z)^{\ell - 1}$ then $W$ can be partitioned into $\ell - 1$ disjoint sets from $\mathcal{F}$. 
		If there exists a monomial $z^{val(\psi(W))} $ in $p(z)^{\ell} = p(z)^{\ell - 1} \star p(z)$ then it is the product of two monomials, say $z^{val(\psi(W_1))}$ in $p(z)^{\ell - 1}$ and $z^{val(\psi(W_2))}$ in $p(z)$ respectively with $W_1 \cap W_2 = \emptyset$. By induction hypothesis, $W_1$ is the disjoint union of $S_1, S_2, \dots, S_{\ell - 1}$ such that $S_i \in \mathcal{F}$ for all $i \in [\ell - 1]$. Also, $W_2$ is in $\mathcal{F}$ and since $W_1 \cap W_2 = \emptyset$, $S_i \cap W_2 = \emptyset$ for each $i$. Therefore, $W$ can be partitioned into sets $S_1, S_2, \dots, S_{\ell - 1}, W_2$ each of which belong to $\mathcal{F}$.  
	\end{proof}

	\noindent Now we are in a position to prove the main theorem of this section. 
	
	\begin{theorem} Given a graph $G$ on $n$ vertices, there is an algorithm which finds its cd-chromatic number in $\mathcal{O}(2^n n^4 \log n)$ time.
	\end{theorem}
	
	\begin{proof} Fix an arbitrary ordering on $V(G)$. With $V(G)$ as the universe, we define the family $\mathcal{F}$ of its subsets as follows. 
		$$\mathcal{F} := \{X \subseteq V(G) |\ X \text{ is an independent set and } \exists \ y \in V(G) \text{ s.t. } X \subseteq N(y)\}$$
		Note that every set in $\mathcal{F}$ is an independent set and there exists a vertex which dominates it. That is, $\mathcal{F}$ is the collection of the possible color classes in any cd-coloring of $G$. Let $p(z)$ be the characteristic polynomial of $\mathcal{F}$. By Lemma~\ref{lemma:generalized-partition}, if there exists a monomial $z^{val(\psi(V(G)))}$ in $p(z)^{\ell}$ then $V(G)$ can be partitioned into $\ell$ sets each belonging to $\mathcal{F}$. Hence the smallest integer $\ell$ for which there exists a monomial $z^{val(\psi(V(G)))}$ in $p(z)^{\ell}$ is $\chi_{cd}(G)$. By Corollary~\ref{cor:running-time}, $p(z)^{\ell}$ can be computed in $\mathcal{O}(2^n n^3 \log n \cdot l)$ time. 
As the cd-chromatic number of a graph is upper bounded by $n$, the claimed running time bound for the algorithm follows.
\end{proof}	
\section{FPT Algorithms for cd-Chromatic Number}
	Determining whether a graph $G$ has cd-chromatic number at most $q$ is \NP-hard on general graphs for $q \ge 4$. This implies that the 
	\textsc{cd-Coloring} problem parameterized by the number of colors is para-\NP-hard on general graphs.  Thus this necessitates the search for special classes of graphs where \textsc{cd-Coloring} 
	is \FPT. In this section we give \FPT\ algorithms for \textsc{cd-Coloring} on chordal graphs and graphs of girth at least $5$. 
	
	We start by proving that \textsc{cd-Coloring} parameterized by the number of colors and treewidth of the graph is \FPT.\ Towards this, we will use Courcelle's powerful theorem which interlinks the fixed parameter tractability of a certain graph property with its expressibility as an MSO formula. We can write many graph theoretical properties as an MSO formula. Following are three examples which we will use in writing an MSO formula to check whether a graph has cd-chromatic number at most $q$. 
	
	\begin{itemize}
		\item  To check whether $V_1, V_2, \dots ,V_q$ is a partition of $V(G)$.
		$${\sf Part}(V_1, V_2, \dots ,V_q) \equiv \forall u \in V(G)[\exists i \in [q] [(u \in V_i) \land (\forall j \in [q][i \neq j \Rightarrow u \not\in V_j)]]]$$
		\item To check whether a given vertex set $V_i$ is an independent  set or not.
		$${\sf IndSet}(V_i) \equiv \forall u \in V_i [\forall v \in V_i [\lnot adj(u, v)]]$$
		\item To check whether given vertex set $V_i$ is dominated by some vertex or not.
		$${\sf Dom}(V_i) \equiv \exists u \in V(G)[\forall v \in V_i[adj(u, v)]]$$
	\end{itemize}
	 We use $\phi(G, q)$ to denote the MSO formula which states that $G$ has cd-chromatic number at most $q$. We use the formulas defined above as macros in $\phi(G, q)$.
	
	\vspace{0.25cm}
	\begin{tabular}{ccl}
		$\phi(G, q)$ & $\equiv$ & $\exists V_1, V_2, \dots, V_q \subseteq V(G)[{\sf Part}(V_1, V_2, \dots, V_q) \land$\\
		& & ${\sf IndSet}(V_1) \land \dots \land {\sf IndSet}(V_q) \land {\sf Dom}(V_1) \land \cdots \land {\sf Dom}(V_q)$]
	\end{tabular}
	
	 It is easy to see that the length of $\phi(G, q)$ is upper bounded by a linear function of $q$. By applying Theorem~\ref{thm:courcelle} we obtain the following result.
	
	\begin{theorem} \label{thm:general-graph-fpt} \textsc{cd-Coloring}  parameterized by the  number of colors and the treewidth of the input graph is \FPT.
	\end{theorem}
	
	\subsection{Chordal Graphs}
	
As the graph gets more structured, we expect many \NP-hard problems to get \emph{easier} in some sense on the restricted class of graphs having that structure.
For example, \textsc{Chromatic-Coloring} is \NP-hard on general graphs but it is polynomial time solvable on chordal graphs.
However, \textsc{cd-Coloring}  is \NP-hard even on the chordal graphs and we show that it is \FPT\ when parameterized by the number of colors on chordal graphs.  
	
\begin{theorem} \textsc{cd-Coloring} parameterized by the number of colors is \FPT\ on chordal graphs.
\end{theorem}
\begin{proof} For a chordal graph $G$, $\tw(G) = \omega(G) - 1$ where $\omega(G)$ is the size of a maximum clique in $G$ \cite{golumbic2004algorithmic}. 
Since, a cd-coloring is also a proper coloring, no two vertices in a clique can be in the same color class. Thus, if $\omega(G) \geq k$ then we can conclude that $(G, k)$ is NO instance of \textsc{cd-Coloring}.  Otherwise, $\omega(G) \le k$ which implies that $\tw(G) \le k$. This bound and Theorem~\ref{thm:general-graph-fpt} imply that \textsc{cd-Coloring} parameterized by the number of colors is \FPT\ on chordal graphs. 
	\end{proof}
	
	\subsection{Graphs with girth at least $5$}
	In this section, we show that \textsc{cd-coloring} on graphs of girth at least five is \FPT\ with respect to the solution size as the parameter. By Observation~\ref{obs:graph-connected}, we can assume that the input graph $G$ is connected. We can define cd-coloring of a connected graph as a proper coloring such that every color class is contained in the open neighbourhood of some vertex. In other words, we do not allow a vertex to dominate itself. One can verify that the two definitions of cd-coloring are identical on connected graphs. We now define the notion of a {\em total-dominating set} of a graph $G$.  A set $S\subseteq V(G)$ is called a {\em total-dominating set} if $V(G)=\bigcup_{v\in S} N(v)$. That is, for every vertex $v\in V(G)$, there exists a vertex $u\in S$, $u\neq v$, such that $v\in N(u)$. 
Our interest in total-dominating set is because of its relation to cd-coloring in graphs that do not contain triangles, that is, graphs of girth at least 4.
In particular, we need the following lemma.
The first proof of this has appeared in \cite{merouane2015dominated}.
For the sake of completeness, we present a proof here.

\begin{lemma}[Theorem~$4$ in \cite{merouane2015dominated}]\label{lemma:mim-TDS-CDcol} If $g(G) \ge 4$, then the size of a minimum total dominating set is equal $\chi_{cd}(G)$.
\end{lemma}
\begin{proof} 
Let $\phi$ be a cd-coloring of $G$ that uses $\chi_{cd}(G)$ colors and let $V_1, \dots ,V_q$ be the color classes in this coloring. Then, for every color class $V_i$, there is a vertex $v_i$ such that $V_i \subseteq N(v_i)$. Let $X$ denote the set of these vertices. Then, $X$ has at most $q$ vertices and by definition, it is a total dominating set of $G$. Hence, the size of a minimum total dominating set of a graph is at most the cd-chromatic number of the graph.
		
Suppose $X = \{v_1, v_2, \dots, v_k\}$ is a minimum total dominating set of $G$. We construct a cd-coloring of $G$ using at most $k$ colors. We define the color classes in the following way. Let $V_1 = N(v_1)$ and for $i = 2, \dots, k$, define $V_i = N(v_i) \setminus (V_1 \cup V_2 \cup \dots \cup V_{i-1})$. Note that $V_1, \dots , V_q$ forms a partition of $V(G)$. Since, $g(G)\geq 4$, it follows that each $V_i$ is an independent set. Furthermore, since $X$ is a total dominating set, for each $i \in [k]$, we have a vertex $v_i\in X$ such that $V_i\subseteq N(v_i)$. Hence, this gives a cd-coloring of $G$. Therefore, the cd-chromatic number of a graph is at most the cardinality of a minimum total dominating set. Now the lemma follows by combining the above two inequalities.
\end{proof}
	
	 Lemma~\ref{lemma:mim-TDS-CDcol} shows that to prove that \textsc{cd-Coloring} is \FPT\ on graphs of girth at least four, it suffices to show that finding a total dominating set of size at most $k$ is \FPT\ on these graphs. This leads to the \textsc{Total Dominating Set} problem.  Given a graph $G$ and an integer $k$, the \textsc{Total Dominating Set} problem asks whether there exists a total dominating set of size at most $k$. Observe that we can test whether $G$ has a total dominating set of size at most $k$ by enumerating all subsets $S$ of $V(G)$ of size at most $k$ and checking whether any of them forms a total-dominating set. This immediately gives an algorithm with running time $n^{\Oh(k)}$ for  \textsc{cd-Coloring} on graphs with girth at least $4$, as the checking part can be done in polynomial time. It is not hard to modify the reduction given in~\cite{raman2008short} to show that \textsc{Total Dominating Set} is $W[2]$ hard on bipartite graphs. Thus, Lemma~\ref{lemma:mim-TDS-CDcol} implies that even \textsc{cd-Coloring} is $W[2]$ hard on bipartite graphs. Hence, if we need to show that  \textsc{cd-Coloring} is \FPT, we must assume that the girth of the input graph is at least $5$. In the rest of this section, we show that \textsc{cd-Coloring} is \FPT\ on graphs with girth at least $5$ by showing that \textsc{Total Dominating Set} is \FPT\ on those graphs. Before proceeding further, we note some simple properties of graphs with girth at least $5$.
	
	\begin{observation} \label{obs:nbd-ind} For a graph $G$, if $g(G) \ge 5$ then for any $v$ in $V(G)$, $N(v)$ is an independent set and for any $u, v$ in $V(G)$, $|N(v) \cap N(u)| \le 1$. 
	\end{observation}
	
	 Raman and Saurabh \cite{raman2008short} defined a variation of \textsc{Set Cover} problem, namely,  \textsc{Bounded Intersection Set Cover}. An input to the problem consists of a universe $\mathcal{U}$, a collection $\mathcal{F}$ of subsets of $\mathcal{U}$ and a positive integer $k$ with the property that for any two $S_i, S_j$ in $\mathcal{F}$, $|S_i \cap S_j| \leq c$ for some constant $c$ and the objective is to check whether  there exists a sub-collection $\mathcal{F}_0$ of $\mathcal{F}$ of size at most $k$ such that $\bigcup_{S \in \mathcal{F}_0} = \mathcal{U}$. In the same paper, the authors proved that the \textsc{Bounded Intersection Set Cover} is \FPT\ when parameterized by the solution size. \textsc{Total Dominating Set} on $(G, k)$ where $G$ has girth at least $5$ can be reduced to \textsc{Bounded Intersection Set Cover} with $\mathcal{U} = V(G)$ and $\mathcal{F} = \{N(v) |\  \forall v \in V(G)\}$. By Observation~\ref{obs:nbd-ind}, we can fix the constant $c$ to be $1$. Hence we have the following lemma.
	
	\begin{lemma} \label{lemma:tds-fpt} On graphs with girth at least $5$, \textsc{Total Dominating Set} is \FPT\ when parameterized by the solution size.
	\end{lemma}
	
	 We now prove that the problem has a polynomial kernel and use it to design another \FPT\ algorithm. 
	
	\begin{lemma} \label{lemma:tds-kernel} \textsc{Total Dominating Set} admits a kernel on $\mathcal{O}(k^3)$ vertices on the class of graphs with girth at least 5.
	\end{lemma}
	\begin{proof} We start the proof with the following claim which says that every high degree vertex should be included in every total dominating set of size at most $k$. 
		
		\begin{claim} In a graph $G$ with $g(G) \ge 5$, if there is a vertex $u$ with degree at least $k + 1$, then any total dominating set of size at most $k$ contains $u$. 
\end{claim}
\begin{proof} Suppose there exists a total dominating set $X$ of $G$ of size at most $k$ which does not contain $u$. Since $N(u)$ (having size at least $k + 1$) is dominated by $X$ and no vertex can dominate itself, by the Pigeon Hole Principle, there exists a vertex, say $w$, in $X$ which is adjacent to at least two vertices, say, $v_1, v_2$ in $N(u)$. This implies that $w,v_1,v_2,u$ form a cycle of length $4$, contradicting the fact that girth of $G$ is at least 5.  
\end{proof}
		
		 Suppose $G$ has a total dominating set of size at most $k$. Construct a tri-partition of $V(G)$ as follows:
		\begin{eqnarray*}
			H & = &  \{u \in V(G) ~|~ |N(u)| \ge k + 1\}; \\  
			J & = &  \{v \in V(G) ~|~ v\notin H,~\exists u \in H \text{ such that } (u,v) \in E(G)\};  \\
			R & = & V(G) \setminus (H \cup J)
		\end{eqnarray*} 
		By the above claim, $H$ is contained in every total dominating set of size at most $k$. Hence, the size of $H$ is upper bounded by $k$. Note that there is no edge between a vertex in $H$ and a vertex in $R$. Thus, $R$ has to be dominated by at most $k$ vertices from $J \cup R$. However, the degree of vertices in $J \cup R$ is at most $k$ and hence $|R| \le \mathcal{O}(k^2)$ and $|J \cap N(R)|$ is upper bounded by $\mathcal{O}(k^3)$. We will now bound the size of $J^\star=J \setminus N(R)$. For that, we first apply the following reduction rule on the vertices in $J^\star$.
		
		\begin{Reduction Rule} \label{rr:nbd-subset} For $u, v \in J^\star$, if $N(u) \cap H \subseteq N(v) \cap H$ then delete $u$.
		\end{Reduction Rule}
		The correctness of this reduction follows from the observation that all the vertices in $J$ have been dominated by the vertices in $H$. The only reason any vertex in $J^\star$ is part of a total dominating set is because that vertex is used to dominate some vertex in $H$. If this is the case then the vertex $u$ in the solution can be replaced by the vertex $v$. In the reverse direction, if $X$ is a total dominating set of $G - \{u\}$ and $|X| \leq k$, then $H \subseteq X$. Hence $u$ is dominated by $x \in X \cap H$ in $G$ too. That is, $X$ is a total dominating set of $G$. 
		
		All that remains is to bound the size of $J^\star$.  We partition $J^\star$ into two sets namely $J_1$ and $J_2$. The set $J_1$ is the set of vertices which are adjacent to exactly one vertex in $H$ whereas  each vertex in $J_2$ is adjacent to at least two vertices in $H$. After exhaustive application of Reduction Rule~\ref{rr:nbd-subset}, no two vertices in $J_1$ can be adjacent to one vertex in $H$ and hence $|J_1| \le |H| \le k$. Any vertex in $J_2$ is adjacent to at least two vertices in $H$. For every vertex $u$ in $J_2$, we assign a pair of vertices in $H$ to which $u$ is adjacent. By Observation~\ref{obs:nbd-ind}, no two vertices in $J_2$ can be assigned to the same pair and hence the size of $J_2$ is upper bounded by $\binom{k}{2} \le k^2$. Combining all the bounds, we get a kernel with $\mathcal{O}(k^3)$ vertices.
	\end{proof}
	
	
	 Combining Lemmas ~\ref{lemma:mim-TDS-CDcol} and \ref{lemma:tds-kernel} we obtain the following theorem.
	
	\begin{theorem} On graphs with girth at least $5$, \textsc{cd-Coloring} admits an algorithm running in $\Oh(2^{\Oh(q^{3})} q^{12} \log q^3)$ time and an $\mathcal{O}(q^3)$ sized vertex kernel, where $q$ is number of colors.
	\end{theorem}

\section{Complexity of CD-Partization}
	In this section, we study the complexity of \cdp. As recognizing graphs with cd-chromatic number at most $q$ is \NP-hard on general graphs for $q \geq 4$, the deletion problem is also \NP-hard on general graphs for such values of $q$. For $q=1$, the problem is trivial as $\chi_{cd}(G)=1$ if and only if $G$ is the graph on one vertex. In this section, we show \NP-hardness for $q \in \{2,3\}$. We remark that $\mathcal{G}=\{G \mid \chi_{cd}(G) \leq q\}$ is not a hereditary graph class and so the generic result of Lewis and Yannakakis \cite{node-del} does not imply the claimed \NP-hardness. 
	\subsection{Para-\NP-hardness in General Graphs}
	Consider the following problem. 
	
	\defdecproblem{Partization}{Graph $G$, integers $k$ and $q$}{Does there exist $S \subseteq V(G)$, $|S| \leq k$, such that $\chi(G-S)\leq q$?}
	
	 Once again if $q$ is fixed, we refer to the problem as \textsc{$q$-Partization}. Observe that the classical \NP-complete problems \textsc{Vertex Cover} \cite{garey} and \textsc{Odd Cycle Transversal} \cite{garey} are \textsc{1-Partization} and \textsc{2-Partization}, respectively. Now, we proceed to show the claimed hardness.
	\begin{theorem}
		\textsc{$q$-\cdp} is \NP-complete for $q \in \{2,3\}$.
	\end{theorem}
	\begin{proof}
		The problem is in \NP\ as determining if the cd-chromatic number of a graph is at most $q \in \{1,2,3\}$ is polynomial-time solvable. Given an instance $(G,k)$ of \textsc{$q$-Partization} where $q \in \{1,2\}$, we construct the instance $(G',k)$ of \textsc{$(q+1)$-\cdp} as follows: $G'$ is obtained from $G$ by adding a new vertex $v$ adjacent to every vertex in $V(G)$ and adding $k+q+2$ new vertices $v_1,\cdots,v_{k+q+2}$ adjacent to $v$. We claim that $G$ has a set of $k$ vertices whose deletion results in a $q$-colorable graph if and only if $G'$ has a set of $k$ vertices whose deletion results in a $(q+1)$-cd-colorable graph. 
		
		Consider a set $S$ of $k$ vertices such that $\chi(G-S) \leq q$. Then, $G'-S$ is $(q+1)$-cd-colorable as a new color can be assigned to $v$ and any of the $q$ colors of $G-S$ can be assigned to $v_1,\cdots,v_{k+q+2}$. The color class containing $v$ is a singleton set. This class is dominated by all vertices in $G'-(S \setminus \{v\})$. Further, $v$ dominates each of the other $q$ color classes as $v$ is a universal vertex in $G'$. 
		
		Conversely, let $S' \subseteq V(G')$ be a minimal set of at most $k$ vertices such that $\chi_{cd}(G'-S') \leq q+1$. Now, if $v \in S'$, then vertices $v_1,\cdots,v_{k+q+2}$ are isolated in $G-\{v\}$ implying that either $|\{v_1,\cdots,v_{k+q+2}\} \cap S'| \geq k+1$ or $\chi_{cd}(G'-S') > q+1$. So, we can assume that $v \notin S'$. Further, as $S'$ is minimal, it follows that $\{v_1,\cdots,v_{k+q+2}\} \cap S'=\emptyset$. Also, as $v$ is a universal vertex in $G'$, we have that $\chi(G-(S' \setminus \{v\})) \leq q$. So, $S'$ is a subset of $V(G)$ of size at most $k$ such that $G-S'$ is $q$-colorable. 
	\end{proof}
	\subsection{\NP-hardness and Fixed-Parameter Tractability in Split Graphs}
	A graph is a {\em split graph}  if its vertex set can be partitioned into a clique and an independent set. As split graphs are perfect (clique number is equal to the chromatic number for every induced subgraph), we have the following observation.
	\begin{obs}
		\label{split-r-col}
		A split graph $G$ is $r$-colorable if and only if $\omega(G) \leq r$. 
	\end{obs}
	The following result is known for the corresponding deletion problem.
	\begin{theorem}[\cite{cor,yan}]
		\textsc{Partization on Split Graphs} is \NP-complete.
	\end{theorem}
	This hardness  was shown by a reduction from \textsc{Set Cover} \cite{garey}. We modify this reduction to show that \cdp\ is \NP-complete on split graphs. The problem is in \NP\ as the cd-chromatic coloring of a split graph can be verified in polynomial time due to the following result.
	
	\begin{theorem} [\cite{caldam16}]
		\label{split-cd}
		If $G$ is a connected split graph $G$, then $\omega(G)=\chi_{cd}(G)$. Furthermore, there is an $\Oh({|V(G)|}^2)$ time algorithm that returns a minimum cd-coloring of $G$. 
	\end{theorem}
	
	\begin{theorem}
		\label{split-hard}
		\textsc{\cdp} on split graphs is \NP-hard.
	\end{theorem}
	\begin{proof}
		Consider a \textsc{Set Cover} instance $(U, \mathcal{F}, k)$ where $U=\{x_1,\cdots,x_n\}$ is a finite set and $\mathcal{F}$ is a family $\{S_1,\cdots,S_m\}$ of subsets of $U$. The problem is to determine if there is a collection of at most $k$ sets in $\mathcal{F}$ such that each element of $U$ is in at least one set of the collection. The corresponding instance of \textsc{$cd$-Partization} is $(G, k'=m-k, q=k+1)$ where $G$ is a split graph on the vertex set $C \cup I \cup \{w_0,w_1,\cdots,w_{k+k'+2}\}$ where $C=\{u_i \mid S_i \in \mathcal{F}\}$ and $I=\{v_i \mid x_i \in U\}$. Also, $(v_i,u_j) \in E(G)$ if and only if $x_i \notin S_j$ and $w_0$ is adjacent to every vertex in $C \cup I \cup \{w_1,\cdots,w_{k+k'+2}\}$. Further, $I \cup \{w_1,\cdots,w_{k+k'+2}\}$ and $C$ induce an independent set and a clique, respectively, in $G$. We claim that a set $\mathcal{F}' \subseteq \mathcal{F}$ of size $k$ is a set cover if and only if $G-S'$ is $q$-cd-colorable where $S'=\{u_i \in C \mid S_i \in \mathcal{F} \setminus \mathcal{F}'\}$ and $|S'|=k'$. 
		
		Consider a set cover $\mathcal{F}' \subset \mathcal{F}$ of size $k$. If there is a clique $Q$ (without loss of generality assume $w_0 \in Q$) of size $k+2$ in $G -S'$, then $Q$ must contain an element $v_i \in I$ that is adjacent to $k$ vertices in $C\setminus S'$. However, since $\mathcal{F}'$ is a set cover, $v_i$ is non-adjacent to at least one $u_j$ in $C\setminus S'$ leading to a contradiction. Thus, $S'$ has a non-empty intersection with every $(k+2)$-clique in $G$. As $G$ is a split graph, it is $(k+1)$-colorable due to Observation \ref{split-r-col}. Further, $G-S'$ is $(k+1)$-cd-colorable as the color class containing $\{w_0\}$ is a singleton set (since it is an universal vertex) which is dominated by itself and the other color classes are dominated by $w_0$. 
		
		Conversely, consider a minimal subset $S'$ of $k'$ vertices such that $G - S'$ is $(k+1)$-cd-colorable. Now, if $w_0 \in S'$, then vertices $w_1,\cdots,u_{k+k'+2}$ are isolated in $G-\{w_0\}$ implying that either $|\{w_1,\cdots,w_{k+k'+2}\} \cap S'| \geq k'+1$ or $\chi_{cd}(G-S') > k+1$. So, we can assume that $w_0 \notin S'$. Further, as $S'$ is minimal, it follows that $\{w_1,\cdots,w_{k+k'+2}\} \cap S'=\emptyset$. Now, all vertices in $S'$ must belong to $C$. If there exists $v_i \in S' \cap I$, there is a clique of size $k+2$ in $G - S'$ as $C$ is a clique. Also, no vertex in $I$ is adjacent to all nodes in $C \setminus S'$ as if there is such a vertex $v_i$ then there is a $(k+2)$-clique in $G - S'$. Thus, every vertex in $I$ is nonadjacent to at least one element in $C \setminus S'$ implying that $\{s_i \in \mathcal{F} \mid u_i \in C \setminus S'\}$ is a set cover of $(U, \mathcal{F})$ of size at most $k$. 
\end{proof}
		
As \textsc{Set Cover} parameterized by solution size is $\W[2]$-hard \cite{fpt-book}, we have the following result.
\begin{corollary}
		\textsc{\cdp} on split graphs parameterized by $q$ is $\W[2]$-hard.
	\end{corollary}
	
	 Now, we show that the problem is $\FPT$ with respect to $q$ and $k$.
	
	\begin{theorem} \cdp\ on split graphs is $\FPT$ with respect to parameters $q$ and $k$.
		Furthermore, the problem does not admit a polynomial kernel unless \containment.
	\end{theorem}
	\begin{proof}
		Compute a maximum clique $Q$ of $G$ in polynomial time. If $|Q| \leq q$, then the input instance is an YES instance as $\chi_{cd}(G) \leq q$ from Theorem \ref{split-cd}. Otherwise, choose an arbitrary subset of size $q+1$ from $Q$. Since any solution contains at least one of the $q+1$ vertices, a straightforward branching algorithm runs in $\Oh^*((q+1)^k)$ time. Now, we move on to the kernelization hardness. \textsc{Set Cover} is known not to admit a polynomial kernel when parameterized by the solution size $k'$ and family size $m$ as combined parameters unless \containment~\cite{fpt-book}.  
		The reduction in Theorem \ref{split-hard} produces instances of  \cdp\ where solution size $k$ is $m-k'$ and $q$ is $k'+1$ implying that $q+k$ is $m+1$. Therefore, an $(q+k)^{\Oh(1)}$ kernel for \cdp\ implies an $m^{\Oh(1)}$ kernel for \textsc{Set Cover} which is unlikely.
		
	\end{proof}
\section{Deletion to 3-cd-Colorable Graphs}
In a graph $G$, an edge $e=(u,v)$ is said to be a dominating edge if $N(u) \cup N(v)=V(G)$.
Let $\overline{N[v]}$ denote the set $V(G) \setminus N[v]$. The following characterization of 3-cd-colorable graphs is known from \cite{caldam16}.

\begin{theorem}[\cite{caldam16}] \label{3-char}
A connected graph $G$ satisfies $\chi_{cd}(G) \leq 3$ if and only if $G$ is one of the following types.\\
(Type 0) $G$ is a graph on at most 3 vertices. \\
(Type 1) $G$ is a bipartite graph with a dominating edge.\\
(Type 2) $G$ has a vertex $v$ such that $G-v$ is a bipartite graph with a dominating edge.\\
(Type 3) $G$ has an ordered pair $(x,y)$ of adjacent vertices such that, 
\vspace{-.15cm}
\begin{itemize}
 \item $V(G)=\{x,y\} \uplus X \uplus Y$, 
 \item $G[X \cup \{y\}]$ is a bipartite graph with at least one edge,
 \item $Y \cup \{x\}$ is an independent set, $Y \cup \{x\} \subseteq N(y)$ and $X \cup \{y\}  \subseteq N(x)$.
\end{itemize}
\vspace{-.15cm}
(Type 4) $G$ has an ordered set $(x,y,z)$ of vertices inducing a triangle such that,
\vspace{-.15cm}
\begin{itemize}
 \item $V(G)=\{x,y,z\} \uplus X \uplus Y \uplus Z$, 
 \item $X \subseteq N(x)$, $Y \subseteq N(y)$ and $Z \subseteq N(z)$,
 \item $X \cup \{y\}$, $Y \cup \{z\}$ and $Z \cup \{x\}$ are independent sets.
\end{itemize}
\vspace{-.15cm}
(Type 5) $G$ has an ordered triple $(x,y,z)$ of vertices such that,
\vspace{-.15cm}
\begin{itemize}
 \item $V(G)=\{x,y\} \uplus X \uplus Y \uplus Z$, 
 \item $z \in X\cup Y$, $(x,y) \notin E(G)$ and $(x,z), (y,z) \in E(G)$,
 \item $X \subseteq N(x)$, $Y \subseteq N(y)$ and $Z \subseteq N(z)$,
 \item $X$, $Y$, $Z \cup \{x\}$ and $Z \cup \{y\}$ are independent sets.
\end{itemize}
\end{theorem}
	
	 We refer to the ordered sets in Types 3, 4 and 5 as dominators. In \cite{caldam16}, they are called as d-pair, cd-triangle and NB-triplet respectively. Now, we proceed to solve \textsc{$3$-\cdp}. Let $G$ be the input graph on $n$ vertices, $m$ edges and $k$ be a positive integer. Consider a set $S \subseteq V(G)$ such that $H=G-S$ is $3$-cd-colorable. Then, $H$ is of one of the types listed in Theorem \ref{3-char}. Before we proceed to describe algorithms for each of these types, we list the following well-known results on \textsc{Vertex Cover} and \textsc{Odd Cycle Transversal} that we use in our algorithms.
	
	\begin{theorem}[\cite{vc-fpt-best}]
		\label{vc-best}
		Given a graph $G$ and a positive integer $k$, there is an algorithm running in $\Oh^*(1.2738^k)$ time that determines if $G$ has a vertex cover of size at most $k$ or not.
	\end{theorem}

	\begin{theorem}[\cite{oct-fpt-best}]
		\label{oct-best}
		Given a graph $G$ and a positive integer $k$, there is an algorithm running in $\Oh^*(2.3146^k)$ time that determines if $G$ has an odd cycle transversal of size at most $k$ or not.
	\end{theorem}
	 Here we use notation $\Oh^*$ to suppress functions which are polynomial in size of input. As we would subsequently show, our algorithms reduce the problem of finding an optimum deletion set into finding appropriate vertex covers and constrained odd cycle transversals. The current best parameterized algorithm for finding a vertex cover can straightaway be used as a subroutine in our algorithm while the one for finding an odd cycle transversal requires the following results. Consider a graph $G$ and let $v$ be a vertex in $G$. Define the graph $G'$ to be the graph obtained from $G$ by deleting $v$ and adding a new vertex $v_{ij}$ for each pair $v_i$, $v_j$ of neighbors of $v$; further $v_{ij}$ is adjacent to $v_i$ and $v_j$. 
	
	\begin{lemma}
		\label{constrained-oct1}
		$G$ has a minimal odd cycle transversal of size at most $k$ that excludes vertex $v$ if and only if $G'$ has a minimal odd cycle transversal of size at most $k$.
	\end{lemma}
	\begin{proof}
		Consider an odd cycle transversal $O$ of $G$ excluding $v$ and let $(X,Y)$ be a bipartition of $G-O$. Without loss of generality, let $v \in X$. Then, every vertex in $N(v)$ is either in $O$ or in $Y$. Thus, $X'=(X \setminus \{v\}) \cup (V(G') \setminus V(G))$ is an independent set in $G'$. Consequently, $(X',Y)$ is a bipartition of $G'-O$ implying that $O$ is an odd cycle transversal of $G'$. Conversely, any odd cycle transversal $O'$ of $G'$ can be modified to one that excludes each vertex in $\{v_{ij} \in V(G') \mid v_i,v_j \in N(v)\}$ without increasing the size since any induced odd cycle through $v_{ij}$ is also an induced odd cycle through $v_i$ and $v_j$. Then, it follows that $O'$ is an odd cycle transversal of $G$ that excludes $v$.  
	\end{proof}
	 Let $P,Q \subseteq V(G)$ be two disjoint sets. Let $G''$ be the graph constructed from $G$ by adding an independent set $I_P$ of $k+1$ new vertices each of which is adjacent to every vertex in $P$ and an independent set $I_Q$ of $k+1$ new vertices each of which is adjacent to every vertex in $Q$. Further, every vertex in $I_P$ is adjacent to every vertex in $I_Q$. 
	
	\begin{lemma}
		\label{constrained-oct2}
		$G$ has a minimal odd cycle transversal $O$ of size at most $k$ such that $G-O$ has a bipartition $(X,Y)$ with $P \subseteq X$ and $Q \subseteq Y$ if and only if $G''$ has a minimal odd cycle transversal of size at most $k$.
	\end{lemma}
	\begin{proof}
		Suppose $G-O$ has a bipartition $(X,Y)$ such that $P \subseteq X$ and $Q \subseteq Y$. Then, $G''-O$ has a bipartition $(X',Y')$ where $X'=X \cup I_Q$ and $Y'=Y \cup I_P$. Thus, $O$ is an odd cycle transversal of $G''$ too. Conversely, consider a minimal odd cycle transversal $O'$ of size $k$ of $G''$. Clearly, $O'$ excludes at least one vertex $a$ from $I_P$ and at least one vertex $b$ from $I_Q$. Consider an arbitrary bipartition $(A,B)$ of $G''-O'$ and let $a \in A$ and $b\in B$. Then, as $O'$ is minimal $I_P \subseteq A$ and $I_Q \subseteq B$. That is, $O' \cap (I_P \cup I_Q) =\emptyset$. Further, as any two vertices $p \in I_P$ and $q \in I_Q$ are adjacent, $I_P \cap V(G''-O') \subseteq A$ and $I_Q \cap V(G''-O') \subseteq B$. Thus, $P \subseteq B$ and $Q \subseteq A$. 
	\end{proof}

	\subsection{Deletion to Types 0, 1 and 2}
	It is trivial to check if $G$ has a solution whose deletion results in a graph $H$ with at most 3 vertices. So, deletion to Type 0 is easy. Now, suppose $H$ is of Type 1. Then, we need to identify an edge of $G$ that is a dominating edge for $H$. We describe an algorithm based on this observation. 
	
	\begin{algorithm}[H]
		\DontPrintSemicolon
		\SetKwFunction{Union}{Union} \SetKwFunction{FindCompress}{FindCompress} 
		\SetKwInOut{Input}{Input}\SetKwInOut{Output}{Output}
		\Input{A graph $G$ and a positive integer $k$.}
		\Output{$S \subseteq V(G)$, $|S| \leq k$ such that $G-S$ is of Type 1 (if one exists).}
		\BlankLine
		
		\nl \For{each edge $(x,y)$ in $G$} 
		{
			Let $X'=N(x)\cap \overline{N[y]}$ and $Y'=N(y)\cap \overline{N[x]}$.\;
			Let $S'$ be $V(G) \setminus (X' \cup Y')$ and decrease $k$ by $|S'|$.\;
			\nl \For{each $k_1$ and $k_2$ such that $k_1+k_2 \leq k$} 
			{
				\nl Compute a vertex cover $S_1$ of $G[X']$ with $|S_1| \leq k_1$ (if one exists).\;
				/* $(X' \setminus S_1) \cup \{y\}$ is an independent set */\;
				\nl Compute a vertex cover $S_2$ of $G[Y']$ with $|S_2| \leq k_2$ (if one exists).\;
				/* $(Y' \setminus S_2) \cup \{x\}$ is an independent set */\;
				\If{$S_1$ and $S_2$ are non-empty sets}
				{\Return{$S' \cup S_1 \cup S_2$ }
				}
			}
		}
		\caption{Deletion-to-Type1$(G,k)$}
		\label{type1}
	\end{algorithm}
	
	\begin{lemma}
		\label{type1-thm}
		Algorithm \ref{type1} runs in $\Oh^*(1.2738^k)$ time. 
	\end{lemma}
	\begin{proof}
		The outer loop (step 1) is executed at most $m$ times (once for each edge) and the inner loop (step 2) is executed at most $k^2$ times. Let $(x,y)$ be an edge in $G$. We need to extend $\{x\}$ and $\{y\}$ into independent sets $Y$ and $X$ respectively, such that $X$ is dominated by $x$ and $Y$ is dominated by $y$. Clearly, neighbors of $x$ and non-neighbors of $y$ cannot be in $Y$. Similarly, neighbors of $y$ and non-neighbors of $x$ cannot be in $X$. No common neighbor of $x$ and $y$ can be in either $X$ or $Y$. Thus, the candidates for $X$ and $Y$ are $X'=N(x)\cap \overline{N[y]}$ and $Y'=N(y)\cap \overline{N[x]}$ respectively. All vertices in $V(G) \setminus (X' \cup Y')$ are in any solution. Let $k'=k-|V(G) \setminus (X' \cup Y')|$. Then, $G$ has a $3$-cd-partization solution $S$ of size at most $k$ such that $G-S$ is of Type 1 with $(u,v)$ as a dominating edge if and only if there exists integers $k_1$, $k_2$ with $k_1+k_2 \leq k'$ such that $G[X']$ has a vertex cover of size at most $k_1$ and $G[Y']$ has a vertex cover of size at most $k_2$. Now, Steps 3 and 4 take $\Oh^*(1.2738^k)$ time from Theorem \ref{vc-best}. Thus, the overall running time is $\Oh^*(1.2738^k)$.  
	\end{proof}
	
	 Suppose $H$ is of Type 2. Then, for each vertex $v$ of $G$, we simply run Algorithm \ref{type1} on $G-\{v\}$ with parameter $k$. 
	
	\begin{algorithm}[H]
		\DontPrintSemicolon
		\SetKwFunction{Union}{Union} \SetKwFunction{FindCompress}{FindCompress} 
		\SetKwInOut{Input}{Input}\SetKwInOut{Output}{Output}
		\Input{A graph $G$ and a positive integer $k$.}
		\Output{$S \subseteq V(G)$, $|S| \leq k$ such that $G-S$ is of Type 2 (if one exists).}
		\BlankLine
		
		\nl \For{each vertex $x$ in $G$} 
		{
			Deletion-to-Type1$(G-\{x\},k)$.\;
		}
		\caption{Deletion-to-Type2$(G,k)$}
		\label{type2}
	\end{algorithm}
	
	\begin{lemma}
		\label{type2-thm}
		Algorithm \ref{type2} runs in $\Oh^*(1.2738^k)$ time. 
	\end{lemma}
	\begin{proof}
		As Algorithm \ref{type2} calls Algorithm \ref{type1} at most $n$ times, its running time is $\Oh^*(1.2738^k)$.  
	\end{proof}

	\subsection{Deletion to Type 3}
	Suppose $H$ is of Type 3 with dominator $(x,y)$. Then, the following holds.
	
	\begin{obs}[\cite{caldam16}]
		$\overline{N_H[x]}$ is an independent set and $\overline{N_H[x]} \subseteq N_H(y)$. Further, $N_H(x)$ induces a bipartite graph with at least one edge.
	\end{obs}
	 This observation leads to the following algorithm.
	
	\begin{algorithm}[H]
		\DontPrintSemicolon
		\SetKwFunction{Union}{Union} \SetKwFunction{FindCompress}{FindCompress} 
		\SetKwInOut{Input}{Input}\SetKwInOut{Output}{Output}
		\Input{A graph $G$ and a positive integer $k$.}
		\Output{$S \subseteq V(G)$, $|S| \leq k$ such that $G-S$ is of Type 3 (if one exists).}
		\BlankLine
		
		\nl \For{each ordered pair $(x,y)$ of adjacent vertices in $G$} 
		{
			Let $Y'= N(y) \cap \overline{N[x]}$ and $X'=N(x)$.\;
			Let $S'$ be $V(G) \setminus (X' \cup Y')$ and decrease $k$ by $|S'|$.\;
			\nl \For{each $k_1$ and $k_2$ such that $k_1+k_2 \leq k$} 
			{
				\nl Compute a vertex cover $S_1$ of $G[Y']$ with $|S_1| \leq k_1$ (if one exists).\;
				/* $(Y' \setminus S_1) \cup \{x\}$ is an independent set */\;
				\nl Compute a minimal odd cycle transversal $S_2$ of at most $k_2$ vertices (if one exists) in $G[X']$ such that $y \notin S_2$.\;
				\If{$S_1$ and $S_2$ are non-empty sets}
				{\Return{$S' \cup S_1 \cup S_2$ }
				}
			}
		}
		\caption{Deletion-to-Type3$(G,k)$}
		\label{type3}
	\end{algorithm}
	
	\begin{lemma}
		\label{type3-thm}
		Algorithm \ref{type3} runs in $\Oh^*(2.3146^k)$ time. 
	\end{lemma}
	\begin{proof}
		The outer loop (step 1) is executed at most $2m$ times (as there are two ordered pairs for each edge) and the inner loop (step 2) is executed at most $k^2$ times. Consider an edge $(x,y)$ in $G$. If $(x,y)$ is a dominator in $H$, then we need to extend $\{x\}$ into an independent set $Y$ that is dominated by $y$ and extend $\{y\}$ into an induced bipartite graph $B$ (with at least one edge) such that $V(B)$ is dominated by $x$. Observe that $Y$ contains only neighbors of $y$ and $V(B)$ contains only neighbors of $x$. Further, a neighbor of $y$ that is not adjacent to $x$ cannot be in $V(B)$ and a neighbor of $y$ that is adjacent to $x$ cannot be in $Y$. Thus, the candidates for $V(B)$ and $Y$ are $X'=N(x)$ and $Y'=N(y) \cap \overline{N[x]}$ respectively. All vertices in $V(G) \setminus (X' \cup Y')$ are in any solution. Let $k'=k-|V(G) \setminus (X' \cup Y')|$. Now, $G$ has a $3$-cd-partization solution $S$ of size at most $k$ such that $G-S$ is of Type 3 with $(x,y)$ as a dominator if and only if there exists integers $k_1$ and $k_2$ with $k_1+k_2 \leq k'$ such that $G[Y']$ has a vertex cover of size at most $k_1$ and $G[X']$ has an odd cycle transversal of size $k_2$ not containing $y$ such that the resultant bipartite graph is non-edgeless. Clearly step 3 takes $\Oh^*(1.2738^k)$ time. For step 4, we need to find a minimal odd cycle transversal that excludes vertex $y$. We construct a graph $G'$ obtained from $G[X']$ by deleting $y$ and adding a new vertex $y_{ij}$ for each pair $y_i$, $y_j$ of neighbors of $y$; further $y_{ij}$ is adjacent to $y_i$ and $y_j$. From Lemma \ref{constrained-oct1}, we have that $G[X']$ has a minimal odd cycle transversal of size at most $k_2$ not containing $y$ if and only if $G'$ has a minimal odd cycle transversal of size at most $k_2$. Now, by using Theorem \ref{oct-best}, it follows that step 4 takes $\Oh^*(2.3146^k)$ time and this gives us the claimed running time of the algorithm.  
	\end{proof}

	\subsection{Deletion to Type 4}
	Suppose $H$ is of Type 4 and has $(x,y,z)$ as a dominator. Then, we have the following observation.
	
	\begin{obs}[\cite{caldam16}]
		$N_H(x) \cap N_H(y) \cap N_H(z)=\emptyset$ and $\overline{N_H[x]} \cap \overline{N_H[y]} \cap \overline{N_H[z]}=\emptyset$. Further, $X=N_H(x) \cap \overline{N_H[y]}$, $Y=N_H(y) \cap \overline{N_H[z]}$ and $Z=N_H(z) \cap \overline{N_H[x]}$.
	\end{obs}
	
	 Now, we have the following algorithm.
	
	\begin{algorithm}[H]
		\DontPrintSemicolon
		\SetKwFunction{Union}{Union} \SetKwFunction{FindCompress}{FindCompress} 
		\SetKwInOut{Input}{Input}\SetKwInOut{Output}{Output}
		\Input{A graph $G$ and a positive integer $k$}
		\Output{$S \subseteq V(G)$, $|S| \leq k$ such that $G-S$ is of Type 4 (if one exists)}
		\BlankLine
		
		\nl \For{each ordered triple $(x,y,z)$ of pairwise adjacent vertices of $G$} 
		{
			Let $X'=N(x) \cap \overline{N[y]}$, $Y'=N(y) \cap \overline{N[z]}$ and $Z'=N(z) \cap \overline{N[x]}$.\;
			Let $S'$ be $V(G) \setminus (X' \cup Y' \cup Z')$ and decrease $k$ by $|S'|$.\;
			\nl \For{each $k_1$, $k_2$ and $k_3$ such that $k_1+k_2+k_3 \leq k$} 
			{
				\nl Compute a vertex cover $S_1$ of $G[X']$ with $|S_1| \leq k_1$ (if one exists).\;
				/* $(X' \setminus S_1) \cup \{y\}$ is an independent set */\;
				\nl Compute a vertex cover $S_2$ of $G[Y']$ with $|S_2| \leq k_2$ (if one exists).\;
				/* $(Y' \setminus S_2) \cup \{z\}$ is an independent set */\;
				\nl Compute a vertex cover $S_3$ of $G[Z']$ with $|S_3| \leq k_3$ (if one exists).\;
				/* $(Z' \setminus S_3) \cup \{x\}$ is an independent set */\;
				\If{$S_1$, $S_2$ and $S_3$ are non-empty sets}
				{\Return{$S' \cup S_1 \cup S_2 \cup S_3$ }
				}
			}
		}
		\caption{Deletion-to-Type4$(G,k)$}
		\label{type4}
	\end{algorithm}
	
	\begin{lemma}
		\label{type4-thm}
		Algorithm \ref{type4} runs in $\Oh^*(1.2738^k)$ time. 
	\end{lemma}
	\begin{proof}
		 The outer loop (step 1) is executed at most $n^3$ times and the inner loop (step 2) is executed at most $k^3$ times. Consider a triangle $\{x,y,z\}$ in $G$. If $(x,y,z)$ is a dominator in $H$, then we need to extend $\{x\}$, $\{y\}$, $\{z\}$ into independent sets $Y$, $Z$, $X$ dominated by $y$, $z$ and $x$ respectively. Thus, the candidates for $X$, $Y$ and $Z$ are sets $X'=N(x) \cap \overline{N[y]}$, $Y'=N(y) \cap \overline{N[z]}$ and $Z'=N(z) \cap \overline{N[x]}$. All vertices in $S'=V(G) \setminus (X' \cup Y' \cup Z')$ are in any solution. Let $k'=k-|V(G) \setminus S'|$. Then, $G$ has a $3$-cd-partization solution $S$ of size at most $k$ such that $G-S$ is of Type 4 with $(x,y,z)$ as a dominator if and only if there exists integers $k_1$, $k_2$ and $k_3$ with $k_1+k_2+k_3 \leq k'$ such that $G[Y']$ has a vertex cover of size at most $k_1$, $G[Z']$ has a vertex cover of size at most $k_2$ and $G[Z']$ has a vertex cover of size at most $k_3$. Steps 3, 4 and 5 take $\Oh^*(1.2738^k)$ time from Theorem \ref{vc-best} and the overall running time is $\Oh^*(1.2738^k)$.  
\end{proof}
	
\subsection{Deletion to Type 5}
	
Suppose $H$ is of Type 5 and has $(x,y,z)$ as a dominator. Then, we have the following observation.
	
\begin{obs}[\cite{caldam16}]
$\overline{N_H[x]} \cap \overline{N_H[y]}$ is an independent set. Further, $z \in N_H(x) \cup N_H(y)$ and $\overline{N_H[x]} \cap \overline{N_H[y]} \subseteq N_H(z)$. Moreover, in $G-Z$, $N[x] \cup N[y] =V(G-Z)$, $N(x) \setminus N(y) \subseteq X$ and $N(y) \setminus N(x) \subseteq Y$.
\end{obs}
Now, we have the following algorithm.
	
\begin{algorithm}[t]
		\DontPrintSemicolon
		\SetKwFunction{Union}{Union} \SetKwFunction{FindCompress}{FindCompress} 
		\SetKwInOut{Input}{Input}\SetKwInOut{Output}{Output}
		\Input{A graph $G$ and a positive integer $k$}
		\Output{$S \subseteq V(G)$, $|S| \leq k$ such that $G-S$ is of Type 5 (if one exists)}
		\BlankLine
		
		\nl \For{each ordered triple $(x,y,z)$ of vertices of $G$ such that $(x,y) \notin E(G)$ and $(x,z), (y,z) \in E(G)$} 
		{
			Let $Z'$ be the set $N(z) \cap (\overline{N[y]} \cap \overline{N[x]})$.\;
			Let $S'$ be $(N(y) \cap N(z)) \setminus N(x)$.\;
			Let $B'$ be the set  $\{z\} \cup (((N(x) \cup N(y)) \setminus S')$.\;
			Let $S''$ be $V(G) \setminus (Z' \cup B')$ and decrease $k$ by $|S''|$.\;
			\nl \For{each $k_1$ and $k_2$ such that $k_1+k_2 \leq k$} 
			{
				\nl Compute a vertex cover $S_1$ of $G[Z']$ with $|S_1| \leq k_1$ (if one exists).\;
				/* $(Z' \setminus S_1) \cup \{x,y\}$ is an independent set */\;
				\nl Compute a minimal odd cycle transversal $S_2$ of $G[B']$ with $|S_2| \leq k_2$ not containing $z$ (if one exists) such that the resultant bipartite graph has a bipartition $(X,Y)$ such that $X \subseteq N(x)$, $Y \subseteq N(y)$ and $z \in Y$.\;
				\If{$S_1$ and $S_2$ are non-empty sets}
				{\Return{$S'' \cup S_1 \cup S_2$ }
				}
			}
		}
		\caption{Deletion-to-Type5$(G,k)$}
		\label{type5}
	\end{algorithm}
	
	\begin{lemma}
		\label{type5-thm}
		Algorithm \ref{type5} runs in $\Oh^*(2.3146^k)$ time. 
	\end{lemma}
	\begin{proof}
		 Consider an ordered triple $(x,y,z)$ of vertices in $G$. If $(x,y,z)$ is a dominator in $H$, then we need to extend $\{x,y\}$ into an independent set $Z$ that is dominated by $z$ and extend $\{z\}$ into a bipartite graph $B$ with bipartition $(X,Y)$ such that $X$ is dominated by $x$ and $Y$ is dominated by $y$. Thus, the candidates for $Z$ and $V(B)$ are $Z'=N(z) \cap (\overline{N[y]} \cap \overline{N[x]})$ and $B'=\{z\} \cup ((N(x) \cup N(y)) \setminus (N(y)\cap N(z)) \setminus N(x))$ respectively. All vertices in $V(G) \setminus (Z' \cup B')$ are in any solution. Let $k'=k-|V(G) \setminus (Z' \cup B')|$. Then, $G$ has a $3$-cd-partization solution $S$ of size at most $k$ such that $H=G-S$ is of Type 5 with $(x,y,z)$ as a dominator if and only if there exists integers $k_1$ and $k_2$ with $k_1+k_2 \leq k'$ such that $G[Z']$ has a vertex cover of size at most $k_1$ and $G[B']$ has an odd cycle transversal of size $k_2$ not containing $z$ such that the resultant bipartite graph has a bipartition $(X,Y)$ such that $X \subseteq N(x)$, $Y \subseteq N(y)$ and $z \in Y$. Step 3 takes $\Oh^*(1.2738^k)$ time. For step 4, we use Lemmas \ref{constrained-oct1} and \ref{constrained-oct2}. Let $G'$ be the graph obtained from $G[B']$ by deleting $z$ and adding a new vertex $z_{ij}$ for each pair $z_i$, $z_j$ of neighbors of $z$, adjacent to $z_i$ and $z_j$. Now, a minimal odd cycle transversal of $G'$  corresponds to a minimal odd cycle transversal of $G[B']$ not containing $z$. However, we also need the additional constraint that such an odd cycle transversal results in a bipartite graph $B$ which has a bipartition $(X,Y)$ such that $X \subseteq N(x)$ and $Y \subseteq N(y)$. The possible vertices in $B$ are from the set $\{z\} \cup (((N(x) \cup N(y)) \setminus (N(y) \cap N(z)) \setminus N(x))$. The following observations on vertices from this set are easy to verify.
		\begin{itemize}
			\item $N_x=N(x) \setminus (N(y) \cup N(z))$ cannot be dominated by $y$ and $N_y=N(y) \setminus (N(x) \cup N(z))$ cannot be dominated by $x$. 
			\item $N_{zx}=(N(x) \cap N(z)) \setminus N(y)$ and $N_{xyz}=N(x) \cap N(y) \cap N(z)$ cannot be in a part of the bipartition that contains $z$.
		\end{itemize}
		It follows that we need an odd cycle transversal (of size at most $k_2$) of $G[B']$ after deleting which the resultant bipartite graph has a 2-coloring in which any vertex from $P=\{z\} \cup N_y$ receives color 1 and any vertex from $Q=N_x \cup N_{zx} \cup N_{xyz}$ receives color 2. This is achieved by constructing graph $G''$ from $G'$ by adding an independent set $I_P$ of $k_2+1$ new vertices each of which is adjacent to every vertex in $P$ and an independent set $I_Q$ of $k_2+1$ new vertices each of which is adjacent to every vertex in $Q$. Further, every vertex in $I_P$ is adjacent to every vertex in $I_Q$. Now, $G[B']$ has a minimal odd cycle transversal of size at most $k_2$ not containing $z$ such that the resultant bipartite graph has a bipartition $(X,Y)$ such that $X \subseteq N(x)$, $Y \subseteq N(y)$ and $z \in Y$ if and only if $G''$ has a minimal odd cycle transversal of size at most $k_2$. Now, using Theorem \ref{oct-best}, it follows that step 4 takes $\Oh^*(2.3146^k)$ time and the overall running time is dominated (upto polynomial factors) by this computation.
\end{proof}
	
From Lemmata \ref{type1-thm}, \ref{type2-thm}, \ref{type3-thm}, \ref{type4-thm} and \ref{type5-thm}, we have the following result.
\begin{theorem}
Given a graph $G$ and an integer $k$, there is an algorithm that determines if there is a set $S$ of size $k$ whose deletion results in a graph $H$ with $\chi_{cd}(H) \leq 3$ in $\Oh^*(2.3146^k)$ time.
\end{theorem}
	
\section{Concluding Remarks}
In this work, we described exact and \FPT\ algorithms for problems associated with cd-coloring.
We also explored the complexity of finding the cd-chromatic number in graphs of girth at least $5$ and chordal graphs.
On the former graph class, we described a polynomial kernel. 
The kernelization complexity on other graph classes and whether the problem is \FPT\ parameterized by only treewidth are natural directions for further study.
It is also interesting to get an exact function when parameterized by treewidth and the number of colors.	



\bibliographystyle{elsarticle-num} 
\bibliography{ref}





\end{document}